\documentclass[a4paper]{article}

\usepackage[utf8]{inputenc}
\usepackage{microtype}
\usepackage{algorithm}
\usepackage{algorithmic}
\usepackage[T1]{fontenc}
\usepackage{tabularx}
\usepackage{amsmath}
\usepackage{amsthm}
\usepackage{amsfonts}
\usepackage{graphics}
\usepackage{mathtools}
\usepackage{fullpage}

\newcommand{\Order}{\mathrm{O}}
\newcommand{\Exp}{\mathbb{E}}

\newcommand{\Rank}{\textsc{Ranking}}
\newcommand{\KVV}{\textsc{KVV}}
\newcommand{\Greedy}{\textsc{Greedy}}
\newcommand{\eps}{\varepsilon}

\newcommand{\OPT}{\textsc{opt}}
\newcommand{\opt}{\OPT}
\newcommand{\ALG}{\textsc{alg}}
\newcommand{\alg}{\ALG}

\newcommand{\qStr}{{q\text{-}\textsc{sgkh}}}
\newcommand{\cStr}{{c!\text{-}\textsc{sgkh}}}

\newcommand{\mat}{\textsc{mat}}
\newcommand{\OO}{\ensuremath{{O}}}

\newtheorem{theorem}{Theorem}
\newtheorem{corollary}{Corollary}
\newtheorem{lemma}{Lemma}

\newtheorem{definition}{Definition}

\title{On the Power of Advice and Randomization for Online Bipartite Matching \footnote{Research supported in part by the ANR projects ANR-11-BS02-0015,  ANR-15-CE40-0015, ANR-12-BS02-005, by the Icelandic Research Fund grants-of-excellence no.\ 120032011 and 152679-051 and by the European Research Council (ERC) under the European Union’s Horizon 2020 research and innovation programme no.\ 648032.}}

\author{Christoph D\"{u}rr%
\thanks{Sorbonne Universités, UPMC Univ Paris 06, CNRS, LIP6, Paris, France}
\and
Christian Konrad%
\thanks{Reykjavik University, Reykjavik, Iceland}
\and
Marc Renault%
\thanks{IRIF, CNRS,  Universit\'{e} Paris Diderot, Paris, France}}

\begin{document}

\maketitle

\begin{abstract}
While randomized online algorithms have access to a sequence of uniform random bits, deterministic
online algorithms with advice have access to a sequence of {\em advice bits}, i.e., bits that are set
by an all-powerful oracle prior to the processing of the request sequence.
Advice bits are at least as helpful as random bits, but how helpful are they? In this work, we investigate the
power of advice bits and random bits for online maximum bipartite matching (\textsc{MBM}).

The well-known Karp-Vazirani-Vazirani algorithm \cite{kvv90} is an optimal
randomized $(1-\frac{1}{e})$-competitive algorithm for \textsc{MBM} that requires access to
$\Theta(n \log n)$ uniform random bits. We show that $\Omega(\log(\frac{1}{\epsilon}) n)$ advice
bits are necessary and $\Order(\frac{1}{\epsilon^5} n)$ sufficient in order to obtain a
$(1-\epsilon)$-competitive deterministic advice algorithm. Furthermore,
for a large natural class of deterministic advice algorithms, we prove that
$\Omega(\log \log \log n)$ advice bits are required in order to improve on the $\frac{1}{2}$-competitiveness of
the best deterministic online algorithm, while it is known that $\Order(\log n)$
bits are sufficient \cite{bkkk11}.

Last, we give a randomized online algorithm that uses $c n$ random bits, for integers $c \ge 1$,
and a competitive ratio that approaches $1-\frac{1}{e}$ very quickly as $c$ is increasing.
For example if $c = 10$, then the difference between $1-\frac{1}{e}$ and the achieved competitive ratio
is less than $0.0002$.
\end{abstract}

\section{Introduction}
\paragraph*{Online Bipartite Matching}
The maximum bipartite matching problem (\textsc{MBM}) is a well-studied problem in the area of online
algorithms \cite{kvv90,bm08,djk13}. Let $G = (A, B, E)$ be a bipartite graph with
$A = [n]  := \{1, \dots, n \}$ and $B = [m]$, for some integers $n,m$. We
assume $m = \Theta(n)$ allowing bounds to be stated as simple functions of $n$ rather than of $n$ and $m$.
The $A$-vertices together with their incident edges
arrive online, one at a time, in
some adversarial chosen order $\pi: [n] \rightarrow [n]$. Upon arrival of a vertex $a \in A$,
the online algorithm has to irrevocably decide to which of its incident (and yet unmatched) $B$-vertices
it should be matched.
The considered quality measure is the well-established {\em competitive ratio}
\cite{st85}, where
the performance of an online algorithm is compared to the performance of the best
offline algorithm: A randomized online algorithm $\textbf{A}$ for \textsc{MBM}
is $c$-competitive if the matching $M$ output by $\textbf{A}$ is such that
$\mathbb{E} |M| \ge c \cdot |M^*|$, where the expectation is taken over the random coin flips,
and $M^*$ is a maximum matching.

In 1990, Karp, Vazirani and Vazirani \cite{kvv90} initiated research on
online \textsc{MBM} and presented a $(1-\frac{1}{e})$-competitive randomized algorithm denoted \KVV. It chooses
a permutation $\sigma: [m] \rightarrow [m]$ of the $B$-vertices uniformly at random and then runs
the algorithm $\Rank(\sigma)$,
which matches each incoming $A$-vertex $a$ to the free incident $B$-vertex $b$ of minimum rank
(i.e., $\sigma(b) < \sigma(c)$ for all free incident vertices $c \neq b$).
If there is no free $B$-vertex, then $a$ remains unmatched.
They showed that no online algorithm has a better competitive ratio than $1-\frac{1}{e}$, implying
that \KVV~is optimal.
For deterministic online algorithms, it is well-known that the \textsc{Greedy} matching algorithm,
which can be seen as running $\Rank(\sigma)$ using a fixed arbitrary $\sigma$, is $\frac{1}{2}$-competitive,
and is optimal for the class of deterministic online algorithms.

\paragraph*{Improving on $1-\frac{1}{e}$} Additional assumptions are needed in order to improve on the
competitive ratio $1-\frac{1}{e}$. For example, Feldman et al.~\cite{fmmm09} introduced the online stochastic matching
problem, where a bipartite graph $G' = (A', B', E')$ and a probability distribution $\mathcal{D}$ is
given to the algorithm. The request sequence then consists of vertices of $A'$ that are
drawn according to $\mathcal{D}$. Feldman et al.\ showed that the additional knowledge can be used to
improve the competitive ratio to $0.67$, which has subsequently
been further improved \cite{bk10,mgs11}. Another example is a work by Mahdian and Yan \cite{my11}, who considered
the classical online bipartite matching problem with a random arrival order of vertices. They analysed the
\textsc{KVV}~algorithm for this situation and proved that it is $0.696$-competitive.

\paragraph*{Online Algorithms with Advice}
It is a common theme in online algorithms to equip an algorithm with additional knowledge that
allows it to narrow down the set of potential future requests and, thus, design algorithms that have
better competitive ratios as compared to algorithms that have no knowledge about the future. Additional
knowledge can be provided in many different ways, e.g.\ access to lookahead \cite{hs92,g95}, probability
distributions about future requests \cite{fmmm09,my11},
or even by giving an isomorphic copy of the
input graph to the algorithm beforehand \cite{h99}. Dobrev et al. \cite{DRP2008} and later Emek et al. \cite{efkr11} first quantified
the amount of additional knowledge (advice) given to an online algorithm in an information theoretic sense.
They showed that a specific problem requires at least $b(n)$ bits of advice, for some function $b$, in order
to achieve optimality  \cite{DRP2008} or in order to achieve a particular competitive ratio \cite{efkr11}.
Advice lower bounds are meaningful in practice as they apply to any potential type of
additional information that could be given to an algorithm.

In the advice model, a computationally all-powerful oracle is given the entire request sequence and
computes an advice string that is provided to the algorithm.
Algorithms with advice are not usually designed with practical considerations in mind but to show a
theoretical limit on what can be done. As such, the algorithms are often impractical due to the nature of
the advice or the complexity in calculating the advice.
However, from a theoretical perspective, advice algorithms are necessary to determine the exact advice
complexity of online problems (how many advice bits are necessary and sufficient) and thus provide limits on
the achievable and more practically relevant lower bounds.

\paragraph*{Our Objective and Previous Results} Our objectives are to determine the advice complexity
of \textsc{MBM} and to investigate the power of random and advice bits for this problem.

A starting point is a result of B\"{o}ckenhauer et al.~\cite{bkkk11}, who gave a method that allows the
transformation of a randomized online algorithm into a deterministic one with advice with a similar
approximation ratio. More precisely, given a randomized online algorithm \textbf{A} for a minimization problem
$\mathcal{P}$ with approximation factor $c$ and possible inputs $\mathcal{I}(n)$ of length $n$,
B\"{o}ckenhauer et al.\ showed that a $(1+\epsilon)c$-competitive deterministic online algorithm \textbf{B}
with $\log n + 2 \log \log n + \log \frac{\log |\mathcal{I}(n)|}{\log (1+\epsilon)}$
bits
\footnote{Throughout the paper, logarithms, where the base is omitted, are implicitly binary logarithms.}
of advice can be deduced from \textbf{A}, for any $\epsilon > 0$,
where $\log$ is the binary logarithm in this paper.
The calculation of the advice and the computations executed by \textbf{B} require exponential time, since
\textbf{A} has to be simulated on all potential inputs $\mathcal{I}(n)$ on all potential
random coin flips.

The technique of Böckenhauer et al. \cite{bkkk11} can also be applied to maximization problems such as \textsc{MBM}\footnote{It is
straightforward to adapt the proof of Theorem~5 of \cite{bkkk11} accordlingly. For completeness, a proof is given in the full
version of this paper.}.
Applied to the $\textsc{KVV}$ algorithm, we obtain the following theorem.  For the proof see appendix~\ref{appendix:boeckenhauer}.

\begin{theorem} \label{thm:bockenhauer}
 There is a deterministic online algorithm with $\Order(\log n)$ bits
 of advice for \textsc{MBM} with competitive ratio $(1-\epsilon) (1-1/e)$, for any $\epsilon > 0$.
\end{theorem}

This result is complemented by a recent result of Mikkelsen~\cite{m15}, who showed that
for {\em repeatable problems} (see \cite{m15} for details) such as \textsc{MBM},
no deterministic online algorithm with advice sub-linear in $n$
has a substantially better competitive ratio than any randomized algorithm without advice.
Thus, using
$\Order(\log n)$ advice bits, a $(1-\epsilon)(1-\frac{1}{e})$-competitive deterministic algorithm can be obtained,
and no algorithm using $o(n)$ advice bits can substantially improve on this result. Furthermore,
Miyazaki~\cite{m14} showed that $\Theta(\log(n!)) = \Theta(n \log n)$ advice bits are necessary
and sufficient in order to compute a maximum matching.

\paragraph*{Our Results on Online Algorithms with Advice}
Consider a deterministic online algorithm with $f(n)$ bits of advice for \textsc{MBM}. Our previous exposition
of related works shows that the ranges $f(n) \in \Omega(\log n) \cap o(n)$ and $f(n) \in \Theta(n \log n)$
are well understood. In this work, we thus focus on the ranges
$f(n) \in o(\log n)$ and $f(n) \in \Omega(n) \cap o(n \log n)$. Our first set of results concerns
$(1-\epsilon)$-competitive deterministic advice algorithms. We show:
\begin{enumerate}
 \item There is a deterministic $(1-\epsilon)$-competitive online algorithm, using
 $\Order(\frac{1}{\epsilon^5} n)$ advice bits for \textsc{MBM}.

 \item Every deterministic $(1-\epsilon)$-competitive online algorithm for \textsc{MBM} uses
 $\Omega(\log (\frac{1}{\epsilon}) n)$ bits of advice.
\end{enumerate}
Our lower bound result is obtained by a reduction from the {\em string guessing game} of B\"{o}ckenhauer et
al.~\cite{bhkkss14}, a problem that is difficult even in the presence of a large number of advice bits.
This technique has repeatedly been applied for obtaining advice lower bounds, e.g.\ \cite{arrs13,GuptaKL13,BoyarKLL14,adkrr15,BoyarKLL16,BianchiBBKP16}.
Our algorithm simulates an augmenting-paths-based algorithm by Eggert et al. \cite{ekms11}, that has originally
been designed for the data streaming model, with the help of advice bits.
It is fundamentally
different to the $\KVV$ algorithm, however, inspired by the simplicity of $\KVV$, we are particularly interested in
the following class of algorithms:
\begin{definition}[\textsc{Ranking}-algorithm]
 An online algorithm \textbf{A} for \textsc{MBM} is called $\Rank$-algorithm if it follows the
 steps: (1) Determine a ranking $\sigma$; (2) Return $\Rank(\sigma)$.
\end{definition}
The \KVV~algorithm is a \Rank-algorithm, where in step (1), the permutation $\sigma$ is chosen uniformly
at random. The algorithm
described in Theorem~\ref{thm:bockenhauer} is a deterministic \Rank-algorithm with $\Order(\log n)$ bits
of advice that computes the permutation $\sigma$ from the available advice bits.
While we cannot answer the question how many advice bits are needed for deterministic online algorithms
in order to obtain a competitive ratio strictly larger than $\frac{1}{2}$ (and thus to improve on \textsc{Greedy}),
we make progress concerning \Rank~algorithms:
\begin{enumerate}
\setcounter{enumi}{2}
 \item Every $\Rank$-algorithm that chooses $\sigma$ from a set of at most
 $C \log \log n$ permutations, for a small constant $C$, has approximation factor
 at most $(\frac{1}{2} + \delta)$, for any $\delta > 0$.
\end{enumerate}
The previous result implies that every $(\frac{1}{2} + \delta)$-competitive deterministic online
\textsc{Ranking}-algorithm requires $\Omega(\log \log \log n)$ advice bits.

Next, since the computation of the advice and the algorithm of Theorem~\ref{thm:bockenhauer} are not efficient,
we are interested in fast and simple \textsc{Ranking} algorithms. We identify a subclass of \textsc{Ranking} algorithms,
denoted \textsc{Category} algorithms, that leads to interesting results, both as deterministic algorithms with advice and
randomized algorithms without advice.

\begin{definition}[\textsc{Category}-algorithm]
 A $\Rank$-algorithm \textbf{A} is called a \textsc{Category}-algorithm if it follows the steps:
 \begin{itemize}
  \item Determine a category function $c: B \rightarrow \{1, 2, 3, \dots, 2^k\}$ for some integer $k \ge 1$ with $2^k < m$;
  \item Let $\sigma_c: [m] \rightarrow [m]$ be the unique permutation of the
$B$-vertices such that for two vertices $b_1, b_2 \in B: $ $\sigma_c(b_1) < \sigma_c(b_2)$ if
and only if $c(b_1) < c(b_2)$ or $(c(b_1) = c(b_2) \mbox{ and } b_1 < b_2)$.
 \item Return $\Rank(\sigma_c)$.
 \end{itemize}
\end{definition}
Categories can be seen as coarsened versions of rankings, where multiple items with adjacent ranks are
grouped into the same category and within a category, the natural ordering by vertex identifier is used.
We prove the following:
\begin{enumerate}
\setcounter{enumi}{3}\item
 There is a deterministic $\frac{3}{5}$-competitive online \textsc{Category}-algorithm,
 using $m$ bits of advice (and thus two categories).
\end{enumerate}
The oracle determines the categories depending on whether a $B$-vertex would be matched by a run of \textsc{Greedy}.
We believe that this type of advice is particularly interesting since it does not require the oracle to compute an
optimal solution.

\paragraph*{Our Results on Randomized Algorithms}
Last, we consider randomized algorithms with limited access to random bits. The \KVV-algorithm selects a
permutation $\sigma$ uniformly at random, and, since there are $m!$ potential permutations,
$\log(m!) = \Theta(m \log m)$ random bits are required in order to obtain a uniform choice.
We are interested in randomized algorithms that employ fewer random bits. We
consider the class of randomized \textsc{Category}-algorithms, where the categories of the $B$-vertices
are chosen uniformly at random. We show:
\begin{enumerate}
\setcounter{enumi}{4} \item
 There is a randomized \textsc{Category}-algorithm using $k m$ random bits
 with approximation factor $1 - \left( \frac{2^k}{2^k + 1}\right)^{2k}$, for any integer $k \ge 1$.
\end{enumerate}
For $k = 1$, the competitive ratio evaluates to $5/9$. It approaches $1-1/e$ very quickly, for example,
for $k = 10$ the absolute difference between the competitive ratio and $1-1/e$ is less than
$0.0002$.
Our analysis is based on the analysis of the \textsc{KVV} algorithm by Birnbaum and Mathieu \cite{bm08}
and uses a result by Konrad et al.~\cite{kmm12}
concerning the performance of the \textsc{Greedy}
algorithm on a randomly sampled subgraph which was originally developed in the context of streaming algorithms.

The results as described above are summarized in Table~\ref{tab:results}.

\begin{table}[ht]
\small
\begin{center}
 \begin{tabularx}{\textwidth}{|llX|}
\hline
Deterministic ratio & \# of advice bits & Description and Authors    \\
\hline
$1$ & $\Theta(n \log n)$ & (Miyazaki \cite{m14}) \\
$1-\epsilon$ & $\Order(\frac{1}{\epsilon^5} n)$ & Application of Eggert et al. \cite{ekms11} (here) \\
$1-\epsilon$ & $\Omega(\log(\frac{1}{\epsilon})n)$ & LB holds for any online algorithm (here) \\
$1-\frac{1}{e} + \epsilon$ & $\Omega(n)$ & LB holds for any online algorithm (Mikkelsen \cite{m15}) \\
$1-\frac{1}{e}$ &  $\Order(\log n)$ & Exp. time \textsc{Ranking}-alg. (Böckenhauer et al. \cite{bkkk11}) \\
$\frac{3}{5}$ & $m$ & \textsc{Category}-algorithm using two categories (here) \\
$\frac{1}{2}+\epsilon$ & $\Omega(\log \log \log n)$ &  LB holds for \textsc{Ranking}-algorithms (here) \\
\hline
Randomized ratio & \# of random bits & Description and Authors    \\
\hline
$1 - \frac{1}{e}$ & $m \log m$ & \textsc{KVV} algorithm (Karp, Vazirani, Vazirani \cite{kvv90}) \\
$1 - \left( \frac{2^k}{2^k + 1}\right)^{2k}$ & $k m$ &   \textsc{Category}-algorithm using $2^k$ categories (here)\\
\hline
\end{tabularx}
\caption{Overview of our results, sorted with decreasing competitiveness. \label{tab:results}}
\end{center}
\end{table}

\paragraph*{Models for Online Algorithms with Advice}
The two main models for online computation with advice are the per-request model of Emek
et at.~\cite{efkr11} and the tape model of B\"{o}ckenhauer et al.~\cite{bkkkm09}.
Both models were inspired by the original model proposed by Dobrev et al.~\cite{DRP2008}.
In the model of Emek et at.~\cite{efkr11}, a bit string of a fixed length is received by the
algorithm with each request for a total amount of advice that is at least linear in the size of
the input. For this work, we use the tape model of B\"{o}ckenhauer et al.~\cite{bkkkm09}, where the
algorithm has access to an infinite advice string that it can access at any time (see
Section~\ref{sec:prelim} for a formal definition), allowing for advice that is sub-linear in
the size of the input. Many online problems have been studied in the setting of online algorithms
with advice (e.g. metrical task system \cite{efkr11}, $k$-server
problem \cite{efkr11,bkkk11,RenaultR15,GuptaKL13}, paging \cite{DRP2008,bkkkm09}, bin packing
problem \cite{RenaultRS15,BoyarKLL16,adkrr15}, knapsack problem \cite{BockKKR14}, reordering buffer management
problem \cite{arrs13}, list update problem \cite{BoyarKLL14}, minimum spanning tree
problem \cite{BianchiBBKP16} and others). Interestingly, a variant of the algorithm with advice for
list update problem of \cite{BoyarKLL14} was used to gain significant improvements in the compression
rates for Burrows-Wheeler transform compression schemes \cite{KamaliL14}.
The information-theoretic lower bound techniques for online algorithms with advice proposed by
Emek et al.~\cite{efkr11} applies to randomized algorithms and uses a reduction to a matching
pennies game (essentially equivalent to the string guessing game). The reduction technique using the string
guessing game of B\"{o}ckenhauer et al.~\cite{bhkkss14} is a refinement specifically for deterministic
algorithms of the techniques of Emek et al.

\paragraph*{Outline} Preliminaries are discussed in Section~\ref{sec:prelim}. Our
$(1-\epsilon)$-competitive algorithm and a related advice lower bound are presented in
Section~\ref{sec:one-minus-eps}. Then, in Section~\ref{sec:lb-rank}, we give the advice lower
bound for $(\frac{1}{2} + \epsilon)$-competitive $\Rank$-algorithms. Last, in Section~\ref{sec:cat-algos}, we consider our randomized \textsc{Category} algorithm
and our $\frac{3}{5}$-competitive advice \textsc{Category} algorithm.

\section{Preliminaries}\label{sec:prelim}
Unless stated otherwise, we consider a bipartite input graph $G=(A, B, E)$ with $A = [n]$ and
$B = [m]$, for integers $m,n$ such that $m = \Theta(n)$.
The neighbourhood of a vertex $v$ in graph $G$ is denoted by $\Gamma_G(v)$.
Let $M$ be a matching in $G$. We denote the set of vertices matched in $M$ by $V(M)$. For a vertex $v \in V(M)$,
$M(v)$ denotes the vertex that is matched to $v$ in $M$. Generally, we write $M^*$ to denote a maximum
matching, i.e., a matching of largest cardinality. For $A' \subseteq A, B' \subseteq B$,
$opt(A', B')$ denotes the size of a maximum matching in $G[A' \cup B']$, the subgraph induced by $A' \cup B'$.

\paragraph*{The \Rank~Algorithm}
Given permutations $\pi: [n] \rightarrow [n]$
and $\sigma: [m] \rightarrow [m]$, we write $\Rank(G, \pi, \sigma)$ to denote the output matching of the
$\Rank$ algorithm when the $A$-vertices arrive in the order given by $\pi$, and the $B$-vertices are ranked
according to $\sigma$. We may write $\Rank(\sigma)$ to denote $\Rank(G, \pi, \sigma)$ if $\pi$ and $G$
are clear from the context.

\paragraph*{The \textsc{Greedy} Matching Algorithm} $\Greedy$ processes the edges
of a graph in arbitrary order and inserts the current edge $e$ into an initially empty matching $M$ if
$M \cup \{e\}$ is a matching. It computes a maximal matching which is of size at least $\frac{1}{2} |M^*|$.

\paragraph*{Category Algorithms} For an integer $k$, let $c: [m] \rightarrow \{1, \dots, 2^k\}$ be an assignment
of categories to the $B$-vertices. Then let $\sigma_c: [m] \rightarrow [m]$ be the unique permutation of the
$B$-vertices such that for two vertices $b_1, b_2 \in B: $ $\sigma_c(b_1) < \sigma_c(b_2)$ if
and only if $c(b_1) < c(b_2)$ or $(c(b_1) = c(b_2) \mbox{ and } b_1 < b_2)$. The previous definition of $\sigma_c$
is based on the natural ordering of the $B$-vertices. This gives a certain stability to the resulting permutation,
since changing the category of a single vertex $b$ does not affect the relative order of the vertices
$B \setminus \{b \}$.

\paragraph*{The Tape Advice Model} For a given request sequence $I$ of length $n$ for a maximization problem,
an \emph{online algorithm with advice} in the \emph{tape advice model} computes the output sequence
$\alg(I,\Phi) = \left<y_1,y_2,\ldots,y_n\right>$, where $y_i$ is a function of the requests from $1$ to
$i$ of $I$ and the infinite binary advice string $\Phi$. Algorithm $\alg$ has an advice complexity of $b(n)$ if,
for all $n$ and any input sequence of length $n$, $\alg$ reads no more than $b(n)$ bits from $\Phi$.

\section{Deterministic $\mathbf{(1-\epsilon)}$-competitive Advice Algorithms} \label{sec:one-minus-eps}

\subsection{Algorithm With $\mathbf{\Order(\frac{1}{\epsilon^5} n)}$ Bits of Advice}

The main idea of our online algorithm is the simulation of an augmenting-paths-based algorithm
with the help of advice bits. We employ the deterministic algorithm of Eggert et al.~\cite{ekms11} that has
been designed for the data streaming model. It computes a $(1-\epsilon)$-approximate
matching, using $\Order(\frac{1}{\epsilon^5})$ passes over the edges of the input graph, where each pass $i$
is used to compute a matching $M_i$ in a subgraph $G_i = G[A_i \cup B_i]$, for
some subsets $A_i \subseteq A$ and $B_i \subseteq B$, using the \textsc{Greedy}
matching algorithm.
In the first pass, $M_1$ is computed in $G$ and thus constitutes a $\frac{1}{2}$-approximation.
Let $M = M_1$.
Then, $\Order(\frac{1}{\epsilon^2})$ phases follow, where in each phase, a set of disjoint augmenting
paths is computed using $\Order(\frac{1}{\epsilon^3})$ applications of the \textsc{Greedy} matching
algorithm (and thus $\Order(\frac{1}{\epsilon^3})$ passes per phase).
At the end of a phase, $M$ is augmented using the augmenting-paths found in this phase. Upon
termination of the algorithm, $M$ constitutes a $(1-\epsilon)$-approximation (see \cite{ekms11} for
the analysis).

The important property that allows us to translate this algorithm into an online algorithm with advice
is the simple observation that the computed matching $M$ is a subset of $\bigcup_i M_i$. For every $i$,
we encode the vertices
$A_i \subseteq A$ and $B_i \subseteq B$ that constitute the vertices of $G_i$ using $n+m$ advice bits.
Furthermore, for every vertex $a \in A$, we also encode the index $j(a)$
of the matching $M_{j(a)}$ that contains the edge that is incident to $a$ in the final matching $M$
(if $a$ is not matched in $M$, then we set $j(a)=0$). Last, using $\Order(\log n)$ bits, we encode the integers
$n$ and $m$, using a self-delimited encoding. Parameters $n,m$ are required in order to determine the word size
that allows the storage of the indices $j(a)$, and to determine the subgraphs $G_i$. The total number of
advice bits is hence $\Order(\frac{1}{\epsilon^5} (n+m) + \log (\frac{1}{\epsilon^5}) (m) + \log(n)) = \Order(\frac{1}{\epsilon^5} n)$.

After having read the advice bits, our online algorithm computes the $\Order(\frac{1}{\epsilon^5})$ \textsc{Greedy}
matchings $M_i$ simultaneously in the background while receiving the requests. Upon arrival of an
$a \in A$, we match it to the $b \in B$ such that $ab \in M_{j(a)}$ incident to $a$ if $j(a) \ge 1$,
and we leave it unmatched if $j(a) = 0$. We thus obtain the following theorem:

\begin{theorem} \label{thm:ub-advice}
 For every $\epsilon > 0$, there is a $(1-\epsilon)$-competitive deterministic online algorithm for \textsc{MBM}
 that uses $\Order(\frac{1}{\epsilon^5} n)$ bits of advice.
\end{theorem}

\subsection{$\mathbf{\Omega(\log(\frac{1}{\epsilon})n)}$ Advice Lower Bound}
We complement the advice algorithm of the previous section with an $\Omega(\log (\frac{1}{\epsilon}) n)$
advice lower bound for $(1-\epsilon)$-competitive deterministic advice algorithms.
To show this, we make use of the lower bound techniques
of \cite{bhkkss14} using the {\em string guessing game}, which is defined as follows.

\begin{definition}{$\qStr$~\textup{\cite{bhkkss14}}.} The {\em string guessing problem with known history}
over an alphabet $\Sigma$ of size $q \ge 2$ ($\qStr$) is an online minimization problem. The input consists
of $n$ and a request sequence $\sigma = r_1,\ldots,r_n$ of the characters, in order, of an $n$ length string.
An online algorithm $A$ outputs a sequence $a_1,\ldots,a_n$ such that $a_i = f_i(n,r_1,\ldots,r_{i-1}) \in \Sigma$
for some computable function $f_i$.  An important aspect of this problem is that the algorithm needs to produce its output character \emph{before} the corresponding request: request $r_i$ is revealed immediately after the algorithm outputs $a_i$. The cost of $A$ is the Hamming distance between $a_1,\ldots,a_n$ and $r_1,\ldots,r_n$.
\end{definition}

In \cite{bhkkss14}, the following lower bound on the number of advice bits is shown for $\qStr$.

\begin{theorem}{\textup{\cite{bhkkss14}}}\label{thm:BockThm}
  Consider an input string of length $n$ for $\qStr$. The minimum number of advice bits for any
  deterministic online algorithm that is correct for more than $\alpha n$ characters,
  for $\frac{1}{q} \le \alpha < 1$, is
  $((1 - H_q(1-\alpha))\log_2 q) n$, where $H_q(p) = p\log_q(q-1)-p\log_qp-(1-p)\log_q(1-p)$ is the
  $q$-ary entropy function.
\end{theorem}

First, we define a sub-graph that is used in the construction of the lower bound sequence.

\begin{definition}
 A bipartite graph is {\em $c$-semi complete}, if it is isomorphic to
 $G = (A, B, E)$ with $A = \{a_1, \dots, a_c\}, B = \{b_1, \dots, b_c\},$ and
 $E = \{ a_i, b_j \, : \, j \ge i \}$.
\end{definition}

The following lemma presents the reduction from $\qStr$ to \textsc{MBM}.

\begin{lemma}\label{lem:reduceMat}
  For an integer $c \ge 3$, suppose that there is a deterministic $\rho$-competitive online algorithm for
  \textsc{MBM}, using $bn$ bits of advice, where $1 - \frac{1}{c} + \frac{1}{c!} \le \rho < 1$.
  Then, there exists a deterministic algorithm for $\cStr$, using $cbn$ bits of advice, that is correct
  for at least $(1 - (1 - \rho)c)n$ characters of the $n$-length string.
\end{lemma}

\begin{proof}
  Let $\alg_\mat$ be a deterministic $\rho$-competitive online algorithm for \textsc{MBM},
  using $bn$ bits of advice, with $1 - \frac{1}{c} + \frac{1}{c!} \le \rho$, for an
  integer $c \ge 3$. We will present an algorithm $\alg_\cStr$ that, in an online manner,
  will generate a request sequence $I_\mat$ based on its input, $I$ (of length $n$), that can be processed
  by $\alg_\mat$. Further, the advice received by $\alg_\cStr$ will be the advice that
  $\alg_\mat$ requires for $I_\mat$. As shown below, the length of $I_\mat$ is $cn$, hence
  $\alg_\cStr$ requires $cbn$ bits of advice. The solution produced by $\alg_\mat$ on
  $I_\mat$ will define the output produced by $\alg_\cStr$.

  Suppose first that the entire input sequence $I$ is known in advance (we will argue later how to get around
  this assumption). Let $\Pi$ be an enumeration of all the permutations of length $c$,
and let $g:\Sigma \to \{1,\ldots,c!\}$ be a bijection between $\Sigma$, the alphabet of the $\cStr$ problem,
and an index of a permutation in $\Pi$. The request sequence $I_\mat$ has a length of $cn$,
consisting of $n$ distinct $c$-semi-complete graphs, where each graph is based on a request of $I$.
That is, for each request $r_i$ in $I$, we append $c$ requests to $I_\mat$ that correspond to
the $A$-vertices of a $c$-semi-complete graph, where the indices of the $B$-vertices
are permuted according to the permutation $\Pi[g(r_i)]$.

Since $I$ is not known in advance,
we must construct $I_\mat$ in an online manner while predicting the requests $r_j$. For each
request $r_j$, the procedure is as follows:

Let $I^{j-1}_\mat$ be the $c(j-1)$-length prefix of $I_\mat$. Note that when predicting
request $r_j$, requests $r_1,\ldots,r_{j-1}$ have already been revealed, and
$I^{j-1}_\mat$ can thus be constructed.
The algorithm $\alg_\cStr$ simulates $\alg_\mat$ on $I^{j-1}_\mat$ followed by
another $c$-semi-complete graph $G_j = (A_j,B_j,E_j)$ such that, for $1 \le k \le c$,
when vertex $a_k \in A_j$ is revealed, the $B$-vertices incident to $a_k$ correspond exactly
to the unmatched $B$-vertices of $B_j$ in the current matching of $\alg_\mat$. By construction,
$\alg_\mat$ computes a perfect matching in $G_j$.
The computed perfect matching corresponds to a permutation $\pi$ at some index $z$ of $\Pi$, and algorithm
$\alg_\cStr$ outputs $g^{-1}(z)$ as a prediction for $r_j$.

Consider a run of $\alg_\mat$ on $I_\mat$. If $\alg_\mat$ computes
a perfect matching on the $j$th semi-complete graph, then our algorithm predicted $r_j$ correctly.
Similarly, if this matching is not perfect, then
our algorithm failed to predict $r_j$. Let $\nu$ be the total number of imperfect matchings,
let $\alg_\mat(I_\mat)$ denote the matching computed by $\alg_\mat$ on $I_\mat$, and let
$\opt(I_\mat)$ denote a perfect matching in the graph given by $I_\mat$. Then:
\begin{align*}
  |\alg_\mat(I_\mat)| \le |\opt(I_\mat)| - \nu \iff \nu \le |\opt(I_\mat)| - \rho\cdot|\opt(I_\mat)| = (1-\rho)cn ~.
\end{align*}
\end{proof}

We prove now the main lower bound result of this section.
\begin{theorem}\label{thm:lbEps}
  For an integer $c \ge 3$, any deterministic online algorithm with advice for \textsc{MBM}
  requires at least $\left(\frac{(1 - H_q(1-\alpha))}{2}\log c\right)n$ bits of advice to be
  $\rho$-competitive for $1 - \frac{1}{c} + \frac{1}{c!} \le \rho < 1$, where $H_q$ is the $q$-ary
  entropy function and $\alpha = 1 - (1 - \rho)c$.
\end{theorem}

\begin{proof}
  For $1 - \frac{1}{c} + \frac{1}{c!} \le \rho < 1$, let $\alg_\mat$ be a deterministic $\rho$-competitive
  online algorithm for \textsc{MBM}, using $bn$ bits of advice.
  By Lemma~\ref{lem:reduceMat}, there exists an algorithm for $\cStr$ that uses
  $cbn$ bits of advice and is correct for at least $\alpha n$ characters of the $n$-length
  input string. The bounds on $\rho$ and $c$ imply
  $1/(c!) \le \alpha \le 1$. Thus, Theorem~\ref{thm:BockThm} implies
  $cbn \ge ((1 - H_q(1-\alpha))\log (c!) )n$ and, hence,
  $$
    b \ge \frac{(1 - H_q(1-\alpha))}{c}\log (c!) \ge \frac{(1 - H_q(1-\alpha))}{2}\log c \text{, as $c! \ge c^{c/2}$.}
  $$
\end{proof}

Setting $\eps = 1/(2c) < 1/c - 1/(c!)$ for all $c \ge 3$, we get the following corollary. Note that,
as $\rho$ approaches $1$ from below, $\alpha$ also approaches $1$ from below and $H_q(1-\alpha)$ approaches $0$.

\begin{corollary} \label{cor:lb-advice}
  For any $0 < \eps \le 1/6$, any $(1 - \eps)$-competitive deterministic online algorithm with advice
  for \textsc{MBM} requires $\OO(\log(\frac{1}{\epsilon}) n)$ bits of advice.
\end{corollary}

\section{Advice Lower Bound for \textsc{Ranking} Algorithms} \label{sec:lb-rank}
Let $\sigma_1, \dots, \sigma_k: [n] \rightarrow [n]$ be rankings.
We will show that there is a $2n$-vertex graph $G = (A, B, E)$ and an arrival order $\pi$ such that
$|\Rank(G, \pi, \sigma_i)| \le n(\frac{1}{2}+\epsilon) + o(n)$, for every $\sigma_i$ and every constant
$\epsilon > 0$, while $G$ contains a perfect matching. Furthermore, the construction is such that
$k \in \Omega (\log \log n)$.

The key property required for our lower bound is the fact that we can partition the set of $B$-vertices
into disjoint subsets $B_1, \dots, B_q$, each of large enough size, such that for every
$B_i$ with $B_i = \{b_1, \dots, b_p \}$ and $b_1 < b_2 < \dots < b_p$, the sequence $(\sigma_j(b_i))_i$ is
monotonic, for every $1 \le j \le k$. In other words, the ranks of the nodes $b_1, \dots, b_p$ appear
in the rankings $\sigma_i$ in either increasing or decreasing order. For each set $B_i$,
we will construct a vertex-disjoint subgraph $G_i$ on which $\Rank$ computes a matching that is close to
a $\frac{1}{2}$-approximation. The subgraphs $G_i$ are based on graph $H_z$ that we define next.

\paragraph*{Construction of $H_z$}
We construct now graph $H_z = (U, V, F)$ with $U = V = [z]$, for some even integer $z$,
on which {\Rank} computes a matching that is close to a $\frac{1}{2}$-approximation, provided that the $V$
vertices are ranked in either increasing or decreasing order.

Let $U = \{u_1, \dots u_z \}$ be so that $u_i$ arrives before $u_{i+1}$ in $\pi$.
Let $V = \{v_1, \dots v_z \}$ be so that $v_i < v_{i+1}$ (which implies $v_i = i$). Then, for $1 \le i \le z/2$
we define $\Gamma_{H_z}(u_i) = \{v_{2i-1}, v_{2i}, v_{2i+1} \}$, and for $z/2 < i \le z$
we define $\Gamma_{H_z}(u_i) = \{ v_{2i-z-1} \}$.
The graph $H_8$ is illustrated in Figure~\ref{fig:lb-ranking}. It has the following properties:
\begin{enumerate}
 \item If the sequence $(\sigma_i(b_j))_j$ is increasing, then $|\Rank(H_z, \pi, \sigma_i)| = z/2$.
 \item If the sequence $(\sigma_i(b_j))_j$ is decreasing, then $|\Rank(H_z, \pi, \sigma_i)| = z/2+1$.
 \item $H_z$ has a perfect matching (of size $z$).
\end{enumerate}

\begin{figure}
{\small
\begin{center}
\includegraphics[height=2.6cm]{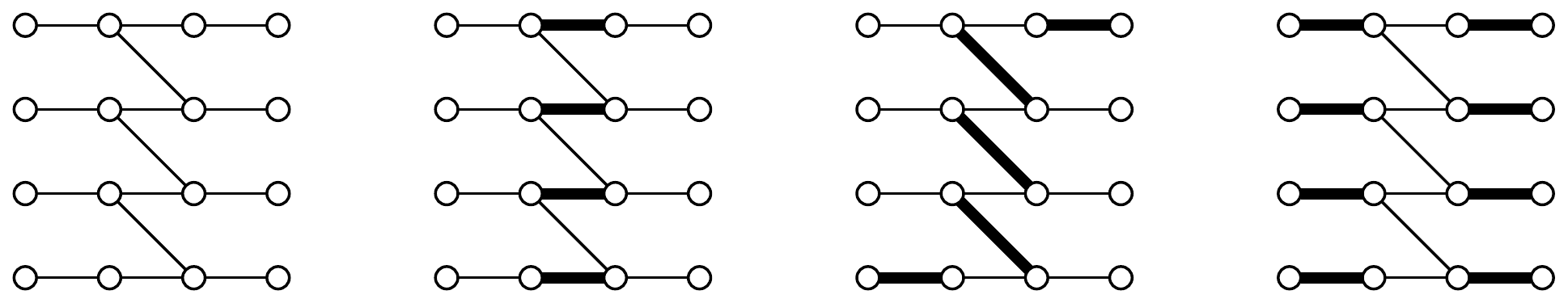}

\vspace{-3.03cm}
\textbf{V} \hspace{0.25cm} \textbf{U} \hspace{0.25cm} \textbf{V} \hspace{0.25cm} \textbf{U} \hspace{10.6cm} $ $

\vspace{0.23cm}

$v_2$ \hspace{0.25cm} $u_1$ \hspace{0.30cm} $v_1$ \hspace{0.25cm} $u_5$ \hspace{10.6cm} $ $

\vspace{0.32cm}

$v_4$ \hspace{0.25cm} $u_2$ \hspace{0.30cm} $v_3$ \hspace{0.25cm} $u_6$ \hspace{10.6cm} $ $

\vspace{0.32cm}

$v_6$ \hspace{0.25cm} $u_3$ \hspace{0.30cm} $v_5$ \hspace{0.25cm} $u_7$ \hspace{10.6cm} $ $

\vspace{0.32cm}

$v_8$ \hspace{0.25cm} $u_4$ \hspace{0.30cm} $v_7$ \hspace{0.25cm} $u_8$ \hspace{10.6cm} $ $

\vspace{-0.1cm}

$ $\hspace{3.8cm} \begin{minipage}{2cm}
\textsc{Ranking}: increasing ranks
\end{minipage} \hspace{1.4cm}
\begin{minipage}{2cm}
\textsc{Ranking}: decreasing ranks
\end{minipage} \hspace{1.2cm}
\begin{minipage}{2.5cm}
perfect matching

$ $
\end{minipage}

\end{center}
}
\caption{Left: $U$-vertices arrive in order $u_1, u_2, \dots, u_8$.
'\textsc{Ranking}: increasing ranks' shows the resulting matching when
$\sigma(v_1) < \sigma(v_2) < \dots < \sigma(v_8)$.
'\textsc{Ranking}: decreasing ranks' shows the resulting matching when
$\sigma(v_1) > \sigma(v_2) > \dots > \sigma(v_8)$.
Right: Perfect matching.
\label{fig:lb-ranking}}
\end{figure}

\paragraph*{Lower Bound Proof}
We prove first that we can appropriately partition the $B$-vertices that allow us to define the
graphs $G_i$. Our prove relies on the well-known Erd\H{o}s-Szekeres theorem \cite{es87} that we state in the
form we need first.

\begin{theorem}[Erd\H{o}s-Szekeres \cite{es87}] \label{thm:es}
 Every sequence of distinct integers of length $n$ contains a monotonic (either increasing or
 decreasing) subsequence of length $\lceil \sqrt{n} \rceil$.
\end{theorem}

\begin{lemma} \label{lem:decomp}
Let $\epsilon > 0$ be an arbitrary small constant. Then for any $k$ permutations
$\sigma_1, \dots, \sigma_k: [n] \rightarrow [n]$ with
$k \le \log \log n - \log \log \frac{1}{\epsilon} -2$,
there is a partition of $B = [n]$ into subsets $C, B_1, B_2, \dots$ such that:
\begin{enumerate}
 \item $|B_i| \ge 1/\epsilon$ for every $i$,
 \item $|C| \le \sqrt{n}$,
 \item For every $B_i = \{b_1, \dots, b_p \}$ with $b_1 < b_2 < \dots < b_p$, and every $\sigma_j$, the sequence
$(\sigma_j(b_l))_l$ is monotonic.
\end{enumerate}
\end{lemma}
\begin{proof}
 Let $S = B$. We iteratively remove subsets $B_i$ from $S$ until $|S| \le \sqrt{n}$. The remaining elements then define
 set $C$. Thus, by construction, Item~2 is fulfilled.

 Suppose that we have already defined sets $B_1, \dots, B_i$. We show how to obtain set $B_{i+1}$.
 Let $S = B \setminus \bigcup_{j=1}^i B_j$ ($S=B$ if $i=0$). Note that $|S| \ge \sqrt{n}$.
 By Theorem~\ref{thm:es}, there is a
 subset $B'_1 = \{b^1_1, \dots, b^1_{\lceil n^{1/4} \rceil} \} \subseteq S$
 with $b^1_1 < b^1_2 < \dots < b^1_{\lceil n^{1/4} \rceil }$ such that the sequence
 $(\sigma_1(b_i))_{b_i \in B'_1}$ is monotonic. Then, again by Theorem~\ref{thm:es}, there is
 a subset $B'_2 = \{b^2_1, \dots, b^2_{ \lceil n^{1/8} \rceil } \} \subseteq B'_1$ with $b^2_1 < b^2_2 < \dots < b^2_{\lceil n^{1/8} \rceil }$
 such that the sequences $(\sigma_j(b_i))_{b_i \in B'_2}$ are monotonic, for every $j \in \{1,2\}$. Similarly,
 we obtain that there is a subset $B'_{w} = \{b^w_1, \dots,  b^w_{\lceil n^{ (1/2)^{w+1} }\rceil} \} \subseteq B'_{w-1}$ with
 $b^w_1 < b^w_2 < \dots < b_{\lceil n^{ (1/2)^{w+1} } \rceil }$ such that the sequences
 $(\sigma_j(b_i))_{b_i \in B'_w}$ are monotonic, for every $j \in \{1, \dots, w\}$.

 In order to guarantee Item~1, we solve the inequality
 $n^{ (\frac{1}{2})^{w+1}} \ge \frac{1}{\epsilon}$
for $w$, and we obtain $w \le \log \log n - \log \log \frac{1}{\epsilon} -2$. This completes the proof.
\end{proof}

Equipped with the previous lemma, we are ready to prove our lower bound result.
\begin{theorem} \label{thm:lb-ranking}
 Let $\epsilon > 0$ be an arbitrary constant. For any $k$
 permutations $\sigma_1, \dots, \sigma_k: [n] \rightarrow [n]$ with
$k \le \log \log n - \log \log \frac{2}{\epsilon} -2$ and arrival order $\pi: [n] \rightarrow [n]$,
there is a graph $G = (A, B, E)$ such that for every $\sigma_i$:
$$|\Rank(G, \pi, \sigma_i)| \le (\frac{1}{2} + \epsilon) n + o(n), $$
while $G$ contains a perfect matching.
\end{theorem}
\begin{proof}
Let $\epsilon' = \epsilon / 2$.
Let $G = (A, B, E)$ denote the hard instance graph.
Let $C,  B_1, B_2, \dots$ denote the partition of $B$ according to Lemma~\ref{lem:decomp} with
respect to value $\epsilon'$.
Then, partition $A$ into sets $A_0, A_1, \dots$ such that $|A_0| = |C|$ and for
$i \ge 1$,  $|A_i| = |B_i|$. Graph $G$ is the disjoint union of subgraphs
$G_0 = (A_0, C, E_0)$ and $G_i = (A_i, B_i, E_i)$, for $i \ge 1$. Subgraph $G_0$
is an arbitrary graph that contains a perfect matching. If $|B_i|$ is even, then
$G_i$ is an isomorphic copy of $H_i$. If $|B_i|$ is odd, then
$G_i$ is the disjoint union of an isomorphic copy of $H_{i-1}$ and one edge.
Then,
 $$|\Rank(G, \pi, \sigma_i)| \le \sum_{B_i} (|B_i|/2 + 2) + |C| \le n/2 + 2 \epsilon' n + \sqrt{n}. $$
\end{proof}

\section{Category Algorithms} \label{sec:cat-algos}
\subsection{Randomized Category Algorithm} \label{sec:rand-cat}
In this section, we analyse the following randomized $\Rank$-algorithm:
\begin{algorithm}[H]
 \begin{algorithmic}
  \REQUIRE $G = (A, B, E)$, integer parameter $k \ge 1$
  \STATE For every $b \in B: c(b) \gets $ random number in $\{1, 2, 3, \dots, 2^k\}$
  \STATE $\sigma_c \gets $ permutation on $[m]$ such that $\sigma_c(b_1) < \sigma_c(b_2)$ iff
  $\left( c(b_1) < c(b_2) \right)$ or $( c(b_1) = c(b_2) $  and $b_1 < b_2)$, for every $b_1, b_2 \in B$
  \RETURN $\Rank(\sigma_c)$
 \end{algorithmic}
 \caption{Randomized Category Algorithm \label{alg:rand}}
\end{algorithm}

\paragraph*{Considering Graphs with Perfect Matchings}
First, similar to \cite{bm08}, we argue that the worst-case performance ratio of Algorithm~\ref{alg:rand} is obtained
if the input graph contains a perfect matching. It requires the following observation:

\begin{theorem}[Monotonicity \cite{gm08,kvv90}]
 Consider a fixed arrival order $\pi$ and ranking $\sigma$ for an input graph $G=(A, B, E)$.
 Let $H = G \setminus \{v \}$ for some vertex $v \in A \cup B$. Let $\pi', \sigma'$ be the arrival order/ranking
 when restricted to vertices $A \cup B \setminus \{v \}$. Then, $\Rank(G, \pi, \sigma)$ and $\Rank(H, \pi', \sigma')$
 are either identical or differ by a single alternating path starting at $v$.
\end{theorem}
The previous theorem shows that the size of the matching produced by Algorithm~\ref{alg:rand} is monotonic
with respect to vertex removals. Hence, if $H$ is the graph obtained from $G$ by removing all vertices that are
not matched by a maximum matching in $G$, then the performance ratio of $\Rank$ on $H$ cannot be better than on $G$.
We can thus assume that the input graph $G$ has a perfect matching and $|A| = |B| = n$.

\paragraph*{Analysis: General Idea}
Let $B_i = \{ b \in B \, : \, c(b) = i \}$, and denote the matching computed by the algorithm by $M$.
The important quantities to consider for the analysis of Algorithm~\ref{alg:rand} are the probabilities:
$$x_i = \Pr_{b \in B} \left[ b \in V(M)  \, | \, b \in B_i \right],$$
i.e., the probability that a randomly chosen $B$-vertex of category $i$ is matched by the algorithm. Determining
lower bounds for the quantities $x_i$ is enough in order to bound the expected matching size, since
\begin{eqnarray}
\nonumber \, \quad \quad \Exp |M| & = & \sum_{b \in B} \Pr \left[ b \in V(M) \right] = \sum_{b \in B} \sum_{i=1}^{2^k}
\Pr \left[b \in B_i \right] \cdot \Pr \left[ b \in V(M) \, | \, b \in B_i \right] \\
& = & \frac{1}{2^k} \sum_{b \in B} \sum_{i=1}^{2^k} \Pr \left[ b \in V(M) \, | \, b \in B_i \right]
= \frac{n}{2^k} \sum_{i=1}^{2^k} x_i . \label{eqn:299}
\end{eqnarray}
We will first prove a bound on $x_1$ using a previous result of Konrad et al. \cite{kmm12}.
Then, using similar ideas as Birnbaum and Mathieu \cite{bm08}, we will prove inequalities of the form
$x_{i+1} \ge f(x_i, \dots, x_1)$, for some function $f$ which allow us to bound the probabilities
$(x_i)_{i \ge 2}$.

\paragraph*{Bounding $x_1$}
Let $H=(U, V, F)$ be an arbitrary bipartite graph and let $U' \subseteq U$ be a uniform and random sample of $U$
such that a node $u \in U$ is in $U'$ with probability $p$. Konrad et al. showed in \cite{kmm12} that when
running \textsc{Greedy} on the subgraph induced by vertices $U' \cup \Gamma_G(U')$, a relatively large fraction
of the $U'$-vertices will be matched, for any order in which the edges of the input graph
are processed that is independent of the choice of $U'$. More precisely, they prove the following theorem
($\Greedy(H', \omega)$ denotes the output of \textsc{Greedy} on subgraph $H'$ if
edges of $H'$ are considered in the order given by $\omega$):
\begin{theorem}[\cite{kmm12}]
 Let $H = (U, V, F)$ be a bipartite graph, $M^*$ a maximum matching, and let $U' \subseteq U$ be a uniform and independent random sample
 of $U$ such that every vertex belongs to $U'$ with probability $p$, $0 < p \le 1$. Then for any edge arrival
 order $\omega$, $$\Exp |\Greedy(H[U' \cup \Gamma_H(U')], \omega)| \ge \frac{p}{1+p} |M^*|.$$
\end{theorem}
In $\Rank$, the vertices $B_1$ are always preferred over vertices $B \setminus B_1$. Thus,
the matching $M_1 = \{ ab \in M \, | \, b \in B_1 \}$ is identical to the matching obtained
when running $\Rank$ on the subgraph induced by $A \cup B_1$. Since the previous theorem
holds for any edge arrival order (that is independent from the choice of $B'$), we can apply
the theorem (setting $B' = B_1, p = \frac{1}{2^k}$) and we obtain:
$$\Exp |B_1 \cap V(M) | \ge \frac{\frac{1}{2^k}}{1 + \frac{1}{2^k}} n = \frac{1}{2^k + 1} n.$$
Since $\Exp |B_1 \cap V(M) | = \sum_{b \in B} \Pr \left[b \in B_1 \right] \cdot
\Pr \left[ b \in V(M) \, | \, b \in B_1 \right] = \frac{n}{2^k} x_1$, we obtain $x_1 \ge 1 - \frac{1}{2^k + 1}$.

\paragraph*{Bounding $(x_i)_{i \ge 2}$} The key idea of the analysis of Birnbaum and Mathieu for
the \textsc{KVV}-algorithm is the observation that, if a $B$-vertex of rank $i$ is not matched by the
algorithm, then its partner in an optimal matching is matched to a vertex of rank smaller than $i$.
Applied to our algorithm, if a $B$-vertex of category $i$ is not matched, then its optimal
partner $M^*(b)$ is matched to a $B$-vertex that belongs to a category $j \le i$. Thus:
\begin{eqnarray}
\nonumber \,\,& &  1 - x_i = \Pr_{b \in B} \left[ b \notin V(M) \, | \, b \in B_i \right]  = \\
& & \quad \quad \Pr_{b \in B} \left[ b \notin V(M) \text{ and } M^*(b) \text{ matched in $M$ to a $b'$ with $c(b') \le i$} \, | \, b \in B_i \right]. \label{eqn:456}
\end{eqnarray}
The following lemma is similar to a clever argument by Birnbaum and Mathieu \cite{bm08}.
\begin{lemma} \label{lem:birnbaum-mathieu}
\begin{eqnarray}
\nonumber \quad & & \Pr_{b \in B} \left[ b \notin V(M) \text{ and } M^*(b) \text{ matched in $M$ to a $b'$ with $c(b') \le i$} \, | \, b \in B_i \right] \\
& & \quad \quad \quad \quad \le \Pr_{b \in B} \left[ M^*(b) \text{ matched in $M$ to a $b'$ with $c(b') \le i$} \right]. \label{eqn:918}
\end{eqnarray}
\end{lemma}
\begin{proof}
 Let $c$ be uniformly distributed and let $\sigma_c$ be the respective ranking. Pick now a random
 $\tilde{b} \in B$ and create new categories $c'$ such that $c'(\tilde{b}) = i$ and for all
 $b \neq \tilde{b}: c'(b) = c(b)$. Let $\sigma_{c'}$ be the ranking given by $c'$.

 Let $\tilde{a} = M^*(\tilde{b})$. Suppose that in a run of $\Rank(\sigma_{c'})$, $\tilde{a}$ is matched
 to a vertex $d'$ with $c'(d') \le i$ and $\tilde{b}$ remains unmatched. Then, we will show that
 in the run of $\Rank(\sigma_c)$, $\tilde{a}$ is matched to a vertex $d$ with $c(d) \le i$. This implies
 our result.

 First, suppose that $\tilde{b}$ remains unmatched in $\Rank(\sigma_c)$. Then, $\Rank(\sigma_c) = \Rank(\sigma_{c'})$
 and the claim is trivially true. Suppose now that $\tilde{b}$ is matched in $\Rank(\sigma_c)$. Then, similar to
 the argument of \cite{bm08}, it can be seen that $\Rank(\sigma_c)$ and $\Rank(\sigma_{c'})$ differ only by
 one alternating path $b_0, a_1, b_1, a_2, b_2, \dots$ starting at $b_0 = \tilde{b}$ such that for all $i$,
  (1) $a_{i+1} b_i \in \Rank(\sigma_c)$,
  (2) $a_i b_i \in \Rank(\sigma_{c'})$, and
  (3) $\sigma_c(b_i) > \sigma_c(b_{i+1})$.
Property (3) implies $c(b_i) \le c(b_{i+1})$. Thus if the category $\sigma_{c'}$ of the node that $a_i$ is matched
to in $\Rank(\sigma_{c'})$ is $k$, then the category $c$ of the node that $a_i$ is matched to in $\Rank(\sigma_c)$ is
also at most $k$.
\end{proof}

The right side of Inequality~\ref{eqn:918} can be computed explicitly as follows:
\begin{eqnarray}
 \nonumber \Pr_{b \in B} \left[ M^*(b) \text{ matched in $M$ to a $b'$ with $c(b') \le i$} \right] = \Pr_{b \in B} \left[ c(b) \le i \text{ and } b \in V(M) \right] =  \frac{1}{2^k} \sum_{j=1}^{i} x_j. \label{eqn:943}
\end{eqnarray}
This, together with Inequalities~\ref{eqn:456} and \ref{eqn:918}, yields $1-x_i \le \frac{1}{2^k} \sum_{j=1}^{i} x_j$.
We obtain:
\begin{theorem}\label{thm:rand}
 Let $k \ge 1$ be an integer. Then Algorithm~\ref{alg:rand} is a randomized online algorithm for
 \textsc{MBM} with competitive ratio
 $1 - \left( \frac{2^k}{2^k + 1}\right)^{2k}$
 that uses $k \cdot m$ random bits.
\end{theorem}
\begin{proof}
 Following \cite{bm08}, the inequality $1-x_i \le \frac{1}{2^k} \sum_{j=1}^{i} x_j$ yields
 $S_i (1 + \frac{1}{2^k}) \ge 1 + S_{i-1}$,
 where $S_i = \sum_{j=1}^i x_i$ and $S_1 = x_1 \ge 1- \frac{1}{2^k + 1}$. According to Equality~\ref{eqn:299},
 we need to bound $S_{2^k}$ from below. Quantity $S_{2^k}$ is minimized if $S_i (1 + \frac{1}{2^k}) = 1 + S_{i-1}$,
 for all $i \ge 2$, which yields
 $$S_{i} = \sum_{j=1}^{i} (1 - \frac{1}{2^k + 1})^j = 2^k  \cdot \left( 1 - \left( \frac{2^k}{2^k + 1}\right)^i \right).$$
 The result follows by plugging $S_{2^k}$ into Equality~\ref{eqn:299}.
\end{proof}

\subsection{Advice Category Algorithm}
Let $\sigma: [m] \rightarrow [m]$ be the identity function, and let $M = \Rank(\sigma)$. It is well-known
that $M$ might be as poor as a $\frac{1}{2}$-approximation. Intuitively, $B$-vertices that are not matched
in $M$ are ranked too high in $\sigma$ and have therefore no chance of being matched. We therefore assign
category $1$ to $B$-vertices that are not matched in $M$, and category $2$ to all other nodes, see
Algorithm~\ref{alg:advice-category}. We will prove that this strategy gives a $\frac{3}{5}$-approximation algorithm.

\begin{algorithm}
 \begin{algorithmic}
  \STATE \textbf{Computation of advice bits}
  \STATE $\sigma \gets $ permutation such that $\sigma(b) = b$, $M_G \gets \Rank(\sigma)$, $M^* \gets $ maximum matching
   \STATE $\forall b \in B: c(b) \gets \begin{cases} 1, \mbox{ if } b \notin V(M), \\ 2, \mbox{ otherwise.} \end{cases}$
  \STATE \textbf{Online Algorithm with Advice} \COMMENT{Function $c$ is provided using $m$ advice bits}
  \STATE $\sigma_c \gets $ permutation on $[m]$ such that $\sigma_c(b_1) < \sigma_c(b_2)$ iff $\left( c(b_1) < c(b_2) \right)$ or $( c(b_1) = c(b_2) $  and $b_1 < b_2)$, for every $b_1, b_2 \in B$
  \RETURN $\Rank(\sigma_c)$
 \end{algorithmic}
\caption{Category-Advice Algorithm \label{alg:advice-category}}
\end{algorithm}

Our analysis requires a property of $\Rank$ that has been previously used, e.g., in \cite{bm08}.
 \begin{lemma}[Upgrading unmatched vertices, Lemma~4 of \cite{bm08}] \label{lem:upgrade}
  Let $\sigma$ be a ranking and let $M = \Rank(\sigma)$. Let $b \in B$ be a vertex that is not matched in $M$.
  Let $\sigma'$ be the ranking obtained from $\sigma$ by changing the rank of $b$ to any rank that is smaller
  than $\sigma(b)$ (and shifting the ranks of other vertices accordingly), and let
  $M' = \Rank(\sigma')$. Then, every vertex $a \in A$ matched in $M$ to a vertex $b \in B$
  is matched in $M'$ to a vertex $b' \in B$ with $\sigma(b') \le \sigma(b)$.
 \end{lemma}

\begin{theorem} \label{thm:adv-cat}
 Alg.~\ref{alg:advice-category} is a $\frac{3}{5}$-competitive online algorithm
 for \textsc{MBM} using $m$ advice bits.
\end{theorem}
\begin{proof}
 Let $M$ denote the matching computed by the algorithm.
 Let $A_2 \subseteq A$, $B_2 \subseteq B$ be the subsets of vertices that are matched in $M_G$.
 Further, let $A_1 = A \setminus A_2$ and $B_1 = B \setminus B_2$ (the vertices not matched in $M_G$).
See Figure~\ref{fig:adv-cat} for an illustration of these quantities.

\begin{figure}[H]
{\small
\begin{center}
\includegraphics[height=2cm]{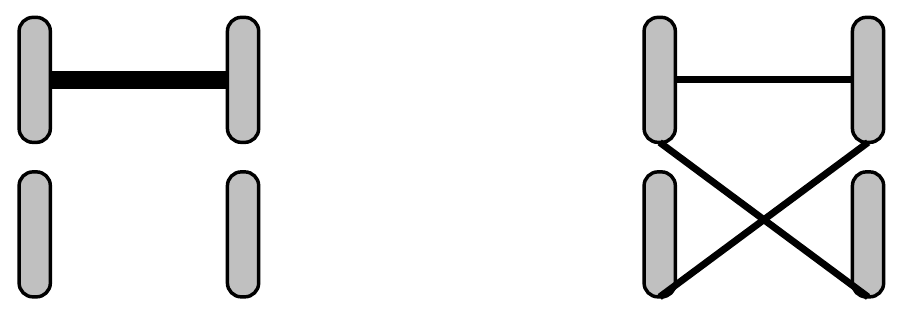}

$\Rank(\sigma)$ \hspace{1.9cm} $\,\Rank(\sigma_c)$

\vspace{-2.3cm}

$A_2$ \hspace{1.5cm} $B_2$ \hspace{1.2cm} $A_2$ \hspace{1.5cm} $B_2$

\vspace{0.6cm}

$A_1$ \hspace{1.5cm} $B_1$ \hspace{1.2cm} $A_1$ \hspace{1.5cm} $B_1$

\vspace{-1.6cm}
$M_G$ \hspace{3.2cm} $M_{22}$

\vspace{0.25cm}
$ $  \hspace{3.3cm} $M_{21}$

\vspace{0.59cm}
$ $  \hspace{3.3cm} $M_{12}$

\end{center}
}
\caption{Quantities employed in the analysis of Algorithm~\ref{alg:advice-category}.}
\label{fig:adv-cat}
\end{figure}

 Then, for $i \in \{1, 2\}$, let $B_i^* = B_i \cap V(M^*)$. Let
 $M_{ij} = \{ab \in M \, | \, a \in A_i \mbox{ and } b \in B_j \}$.
 Then, $M = M_{21} \cup M_{12} \cup M_{22}$ since $M_{11} = \emptyset$ (the input graph does
 not contain any edges between $A_1$ and $B_1$ since otherwise some of them would also be contained in $M_G$).
 This setting is illustrated in Figure~\ref{fig:adv-cat} in the appendix. We will bound now the sizes of $M_{21}, M_{12}$ and $M_{22}$ separately:
 \begin{itemize}
  \item \textit{Bounding $|M_{21}|$.} Since $B_1$-vertices are preferred over $B_2$-vertices in
  $\Rank(\sigma_c)$ and since there are no edges between $A_1$ and $B_1$, $M_{21}$ is a maximal matching between
 $A_2$ and $B_1$. Since $opt(A_2, B_1) = |B_1^*|$, we have $|M_{21}| \ge \frac{1}{2} |B_1^*|.$

 \item \textit{Bounding $|M_{22}|$.} By Lemma~\ref{lem:upgrade}, all $A_2$-vertices are matched in $M$. Thus,
 $|M_{22}| = |A_2| - |M_{21}|.$

 \item \textit{Bounding $|M_{12}|$.} The algorithm finds a maximal matching between $A_1$ and
 $B_2 \setminus B(M_{22})$. Since $opt(A_1, B_2) \ge |A_1^*|$, we have
 $opt(A_1, B_2 \setminus B(M_{22}) ) \ge |A_1^*| - |M_{22}|$,
 and thus $|M_{12}| \ge \frac{1}{2} ( |A_1^*| - |M_{22}| ).$
 \end{itemize}

\noindent We combine the previous bounds and we obtain:
 $$
|M| = |M_{21}| + |M_{22}| + |M_{12}| \ge
|A_2| + \frac{1}{2} ( |A_1^*| - |A_2| + |M_{21}| ) \ge \frac{1}{2} ( |A_1^*| + |A_2| + \frac{1}{2} |B_1^*| ).
 $$
Next, note that $|A_2| \ge |B_1^*|$ and $|A_1^*| + |B_1^*| = |M^*|$. We thus
obtain $|M| \ge \frac{1}{2} |M^*| + \frac{1}{4} |B_1^*|$. Since
$|B_1^*| \ge |M^*| - |M_G|$, we obtain $|M| \ge \frac{3}{4} |M^*| - \frac{1}{4} |M_G|$. Furthermore,
Lemma~\ref{lem:upgrade} implies $|M| \ge |M_G|$, and hence
$|M| \ge \max \{ |M_G|, \frac{3}{4} |M^*| - \frac{1}{4} |M_G| \}$ which is at least $\frac{3}{5} |M^*|$.
\end{proof}

\bibliography{permGuess}

\begin{thebibliography}{10}

\bibitem{arrs13}
Anna Adamaszek, Marc~P. Renault, Adi Ros{\'{e}}n, and Rob van Stee.
\newblock Reordering buffer management with advice.
\newblock In {\em 11th International Workshop on Approximation and Online
  Algorithms (WAOA)}, pages 132--143, September 2013.

\bibitem{adkrr15}
Spyros Angelopoulos, Christoph D{\"u}rr, Shahin Kamali, Marc Renault, and Adi
  Ros{\'e}n.
\newblock Online bin packing with advice of small size.
\newblock In Frank Dehne, J{\"o}rg-R{\"u}diger Sack, and Ulrike Stege, editors,
  {\em Proceedings of the 14th International Symposium on Algorithms and Data
  Structures (WADS)}, pages 40--53. Springer International Publishing, August
  2015.

\bibitem{bk10}
Bahman Bahmani and Michael Kapralov.
\newblock Improved bounds for online stochastic matching.
\newblock In {\em Proceedings of the 18th Annual European Conference on
  Algorithms (ESA)}, pages 170--181. Springer-Verlag, 2010.

\bibitem{BianchiBBKP16}
Maria~Paola Bianchi, Hans{-}Joachim B{\"{o}}ckenhauer, Tatjana
  Br{\"{u}}lisauer, Dennis Komm, and Beatrice Palano.
\newblock Online minimum spanning tree with advice.
\newblock In {\em Proceedings of the 42nd International Conference on Current
  Trends in Theory and Practice of Computer Science (SOFSEM)}, pages 195--207,
  January 2016.

\bibitem{bm08}
Benjamin Birnbaum and Claire Mathieu.
\newblock On-line bipartite matching made simple.
\newblock {\em SIGACT News}, 39(1):80--87, March 2008.

\bibitem{bhkkss14}
Hans{-}Joachim B{\"{o}}ckenhauer, Juraj Hromkovic, Dennis Komm, Sacha Krug,
  Jasmin Smula, and Andreas Sprock.
\newblock The string guessing problem as a method to prove lower bounds on the
  advice complexity.
\newblock {\em Theor. Comput. Sci.}, 554:95--108, 2014.

\bibitem{bkkkm09}
Hans-Joachim B{\"o}ckenhauer, Dennis Komm, Rastislav Kr{\'a}lovi{\v{c}},
  Richard Kr{\'a}lovi{\v{c}}, and Tobias M{\"o}mke.
\newblock On the advice complexity of online problems.
\newblock In Yingfei Dong, Ding-Zhu Du, and Oscar Ibarra, editors, {\em
  Proceedings of the 20th International Symposium on Algorithms and Computation
  (ISAAC)}, pages 331--340. Springer Berlin Heidelberg, December 2009.

\bibitem{BockKKR14}
Hans{-}Joachim B{\"{o}}ckenhauer, Dennis Komm, Richard Kr{\'{a}}lovic, and
  Peter Rossmanith.
\newblock The online knapsack problem: Advice and randomization.
\newblock {\em Theor. Comput. Sci.}, 527:61--72, 2014.

\bibitem{bkkk11}
Hans-Joachim B\"{o}ckenhauer, Dennis Komm, Rastislav Kr\'{a}lovi\u{c}, and
  Richard Kr\'{a}lovi\u{c}.
\newblock On the advice complexity of the k-server problem.
\newblock In {\em Proceedings of the 38th International Colloquium on Automata,
  Languages and Programming (ICALP)}, volume 6755 of {\em Lecture Notes in
  Computer Science}, pages 207--218. Springer Berlin Heidelberg, July 2011.

\bibitem{BoyarKLL14}
Joan Boyar, Shahin Kamali, Kim~S. Larsen, and Alejandro L{\'{o}}pez{-}Ortiz.
\newblock On the list update problem with advice.
\newblock In {\em Proceedings of the 8th International Conference on Language
  and Automata Theory and Applications (LATA)}, pages 210--221, March 2014.

\bibitem{BoyarKLL16}
Joan Boyar, Shahin Kamali, Kim~S. Larsen, and Alejandro L{\'{o}}pez{-}Ortiz.
\newblock Online bin packing with advice.
\newblock {\em Algorithmica}, 74(1):507--527, 2016.

\bibitem{djk13}
Nikhil~R. Devanur, Kamal Jain, and Robert~D. Kleinberg.
\newblock Randomized primal-dual analysis of {RANKING} for online bipartite
  matching.
\newblock In {\em Proceedings of the 24th Annual {ACM-SIAM} Symposium on
  Discrete Algorithms (SODA)}, pages 101--107, January 2013.

\bibitem{DRP2008}
Stefan Dobrev, Rastislav Kr\'{a}lovi\v{c}, and Dana Pardubsk\'{a}.
\newblock How much information about the future is needed?
\newblock In {\em Proceedings of the 34th conference on Current trends in
  theory and practice of computer science (SOFSEM)}, pages 247--258, Berlin,
  Heidelberg, 2008. Springer-Verlag.

\bibitem{ekms11}
Sebastian Eggert, Lasse Kliemann, Peter Munstermann, and Anand Srivastav.
\newblock Bipartite matching in the semi-streaming model.
\newblock {\em Algorithmica}, 63(1):490--508, 2011.

\bibitem{efkr11}
Yuval Emek, Pierre Fraigniaud, Amos Korman, and Adi Ros{\'{e}}n.
\newblock Online computation with advice.
\newblock {\em Theor. Comput. Sci.}, 412(24):2642--2656, 2011.

\bibitem{es87}
Paul Erd{\"o}s and George Szekeres.
\newblock A combinatorial problem in geometry.
\newblock {\em Compositio Mathematica}, 2:463--470, 1935.

\bibitem{fmmm09}
Jon Feldman, Aranyak Mehta, Vahab Mirrokni, and S.~Muthukrishnan.
\newblock Online stochastic matching: Beating 1-1/e.
\newblock In {\em Proceedings of the 2009 50th Annual IEEE Symposium on
  Foundations of Computer Science (FOCS)}, pages 117--126, Washington, DC, USA,
  2009. IEEE Computer Society.

\bibitem{gm08}
Gagan Goel and Aranyak Mehta.
\newblock Online budgeted matching in random input models with applications to
  adwords.
\newblock In {\em Proceedings of the Nineteenth Annual ACM-SIAM Symposium on
  Discrete Algorithms (SODA)}, pages 982--991, Philadelphia, PA, USA, 2008.
  Society for Industrial and Applied Mathematics.

\bibitem{g95}
Edward~F. Grove.
\newblock Online bin packing with lookahead.
\newblock In {\em Proceedings of the Sixth Annual ACM-SIAM Symposium on
  Discrete Algorithms (SODA)}, pages 430--436, Philadelphia, PA, USA, 1995.
  Society for Industrial and Applied Mathematics.

\bibitem{GuptaKL13}
Sushmita Gupta, Shahin Kamali, and Alejandro L{\'{o}}pez{-}Ortiz.
\newblock On advice complexity of the k-server problem under sparse metrics.
\newblock In {\em Proceedings of the 20th International Colloquium on
  Structural Information and Communication Complexity (SIROCCO)}, pages 55--67,
  July 2013.

\bibitem{h99}
Magn\'{u}s~M. Halld\'{o}rsson.
\newblock Online coloring known graphs.
\newblock In {\em Proceedings of the Tenth Annual ACM-SIAM Symposium on
  Discrete Algorithms (SODA)}, pages 917--918, Philadelphia, PA, USA, 1999.
  Society for Industrial and Applied Mathematics.

\bibitem{hs92}
Magn\'{u}s~M. Halld\'{o}rsson and M\'{a}ri\'{o} Szegedy.
\newblock Lower bounds for on-line graph coloring.
\newblock In {\em Proceedings of the Third Annual ACM-SIAM Symposium on
  Discrete Algorithms (SODA)}, pages 211--216, Philadelphia, PA, USA, 1992.
  Society for Industrial and Applied Mathematics.

\bibitem{KamaliL14}
Shahin Kamali and Alejandro L{\'{o}}pez{-}Ortiz.
\newblock Better compression through better list update algorithms.
\newblock In {\em Proceedings of the Data Compression Conference (DCC)}, pages
  372--381, March 2014.

\bibitem{kvv90}
R.~M. Karp, U.~V. Vazirani, and V.~V. Vazirani.
\newblock An optimal algorithm for on-line bipartite matching.
\newblock In {\em Proceedings of the Twenty-second Annual ACM Symposium on
  Theory of Computing (STOC)}, pages 352--358, New York, NY, USA, 1990. ACM.

\bibitem{kmm12}
Christian Konrad, Fr{\'{e}}d{\'{e}}ric Magniez, and Claire Mathieu.
\newblock Maximum matching in semi-streaming with few passes.
\newblock In {\em Approximation, Randomization, and Combinatorial Optimization.
  Algorithms and Techniques}, volume 7408 of {\em Lecture Notes in Computer
  Science}, pages 231--242. Springer Berlin Heidelberg, 2012.

\bibitem{my11}
Mohammad Mahdian and Qiqi Yan.
\newblock Online bipartite matching with random arrivals: An approach based on
  strongly factor-revealing lps.
\newblock In {\em Proceedings of the Forty-third Annual ACM Symposium on Theory
  of Computing (STOC)}, pages 597--606, New York, NY, USA, 2011. ACM.

\bibitem{mgs11}
Vahideh~H. Manshadi, Shayan~Oveis Gharan, and Amin Saberi.
\newblock Online stochastic matching: Online actions based on offline
  statistics.
\newblock In {\em Proceedings of the Twenty-second Annual ACM-SIAM Symposium on
  Discrete Algorithms (SODA)}, pages 1285--1294. SIAM, 2011.

\bibitem{m15}
Jesper~W. Mikkelsen.
\newblock Randomization can be as helpful as a glimpse of the future in online
  computation.
\newblock In {\em Proceedings of the 43rd International Colloquium on Automata,
  Languages, and Programming (ICALP)}, July 2016.

\bibitem{m14}
Shuichi Miyazaki.
\newblock On the advice complexity of online bipartite matching and online
  stable marriage.
\newblock {\em Inf. Process. Lett.}, 114(12):714--717, December 2014.

\bibitem{RenaultR15}
Marc~P. Renault and Adi Ros{\'{e}}n.
\newblock On online algorithms with advice for the k-server problem.
\newblock {\em Theory Comput. Syst.}, 56(1):3--21, 2015.

\bibitem{RenaultRS15}
Marc~P. Renault, Adi Ros{\'{e}}n, and Rob van Stee.
\newblock Online algorithms with advice for bin packing and scheduling
  problems.
\newblock {\em Theor. Comput. Sci.}, 600:155--170, 2015.

\bibitem{st85}
Daniel~D. Sleator and Robert~E. Tarjan.
\newblock Amortized efficiency of list update and paging rules.
\newblock {\em Commun. ACM}, 28(2):202--208, February 1985.

\end{thebibliography}

\newpage
\appendix

\section{The Construction of B\"{o}ckenhauer et al. for Maximization Problems}
\label{appendix:boeckenhauer}

B\"{o}ckenhauer et al. \cite{bkkk11} showed that a deterministic advice algorithm can be obtained
from a randomized algorithm for a minimization problem. We provide a similar theorem for maximization
problems. The proof follows the proof of \cite{bkkk11} and is provided only for completeness of our work.
\label{sec:min-to-max}
\begin{theorem}
  For a maximization online problem $P$, let $\mathcal{I}(n)$ be the set of all possible inputs of
  length $n$. Suppose that $R$ is a randomized algorithm with a expected competitive ratio
  $E(n)$. Then, for any fixed $\eps$, $0 < \eps < 1$, it is possible to construct a deterministic
  algorithm that uses a total advice of size
  $\lceil \log n \rceil + 2 \lceil \log \lceil \log n \rceil \rceil + \left \lceil \log \left\lceil\frac{\log I(n)}{\log (\delta)}\right\rceil \right\rceil$ and that has a competitive ratio of at least
  $(1 - \eps)E(n)$, where $I(n) = |\mathcal{I}(n)|$ and $\delta = \frac{1-E(n)+\eps E(n)}{E(n)}$.
\end{theorem}

\begin{proof}
  Let $r(n)$ be the maximum number of random bits used by $R$ for a sequence of length $n$. Hence, there are $R(n) = 2^{r(n)}$ possible random bit strings. First, we construct a subset, $\mathcal{W}(n)$, of the $R(n)$ random strings with size $\frac{\log I(n)}{\log (1+\eps)}$ such that, for any input $i \in \mathcal{I}(n)$, there exists a random string $w \in \mathcal{W}(n)$ such that $R$, using $w$ as its random string, has a competitive ratio of at least $(1 - \eps)E(n)$.

Let $A$ be a matrix with $I(n)$ rows and $R(n)$ columns, where $A_{i,j}$ is the competitive ratio of $R$ on input $i$ given random string $j$. For any input, the expected competitive ratio of $R$ is at least $E(n)$. So, for any $i$, $1 \le i \le I(n)$,
$$\frac{1}{R(n)}\sum_{j = 1}^{R(n)} A_{i,j} \ge E(n)~\text{and, hence,}~\sum_{i = 1}^{I(n)}\sum_{j = 1}^{R(n)} A_{i,j} \ge I(n) \cdot R(n) \cdot E(n)~.$$
This implies that there exists a random bit string $w$ at some index $j$ such that
\begin{equation}\label{eq:matUpper}
  \sum_{i = 1}^{I(n)} A_{i,j} \ge I(n) \cdot E(n) ~.
\end{equation}
The random bit string $w$ will be included in the set $\mathcal{W}(n)$. Let $\mathcal{H}(n,w)$ be the set of input indexes such that, for input $i \in \mathcal{H}(n,w)$, $A_{i,j} \ge (1 - \eps)E(n)$.  Hence,
\begin{align}\label{eq:matLower}
  \sum_{i = 1}^{I(n)} A_{i,j} &< (I(n)-|\mathcal{H}(n,w)|)(1 - \eps)E(n) + |\mathcal{H}(n,w)|~.
\end{align}

Combining \eqref{eq:matUpper} and \eqref{eq:matLower}, we have that $$|\mathcal{H}(n,w)| > \frac{\eps E(n)}{1 - E(n) + \eps E(n)}I(n) \ge \frac{\eps}{\delta}I(n) \text{~for~} \delta = \frac{1-E(n)+\eps E(n)}{E(n)}.$$
That is, there are at least $\frac{\eps}{\delta}I(n)$ inputs, where $R$, using the random string $w$, has a competitive ratio of at least $(1 - \eps)E(n)$. Each input sequence corresponding to an index in $\mathcal{H}(n,w)$ is said to be \emph{covered} by the bit string $w$.

Let $A'$ be the matrix $A$ with the column corresponding to the string $w$ removed and all the rows at the indexes in $\mathcal{H}(n,w)$ removed. Note that there are $I(n)/\delta$ rows in $A'$, for any row $i$ in $A'$, $\frac{1}{R(n)-1}\sum_{j = 1}^{R(n)-1} A'_{i,j} \ge E(n)$ as, by definition, the contribution of the random string $w$ to a remaining row in $A'$ was less than $E(n)$. Using $A'$, this process can be repeated. After $\lceil \log_{\delta} I(n) \rceil = \left \lceil \frac{\log I(n)}{\log (\delta)} \right \rceil$ iterations of this process, all the input sequence are covered by a bit string in $\mathcal{W}(n)$.

A deterministic algorithm with advice $D$ is defined as follows. Prior to serving a request, $D$ reads $\lceil \log n \rceil + 2 \lceil \log \lceil \log n \rceil \rceil$ advice bits containing the length of the sequence, $n$, encoded as a self-delimited encoding (cf.~\cite{BockKKR14}). Then, $D$ computes the set $\mathcal{W}(n)$ and reads $\left \lceil \log \left\lceil \frac{\log I(n)}{\log (\delta)}\right\rceil\right\rceil$ bits of advice containing an index in $\mathcal{W}(n)$ to the string $w$ that covers the input sequence. Finally, $D$ simulates $R$ on the input sequence with $w$ as the random bit string.
\end{proof}

\end{document}